\documentclass[aps,pra,twocolumn,superscriptaddress,floatfix]{revtex4-2}
\usepackage{graphicx} 
\usepackage{amsmath}
\usepackage{amssymb}
\usepackage{physics}
\usepackage{enumitem}
\usepackage[dvipsnames]{xcolor}
\usepackage{orcidlink}

\usepackage{amsthm}
\newtheorem{lemma}{Lemma}
\newtheorem{theorem}{Theorem}
\newtheorem{proposition}{Proposition}
\newtheorem{corollary}{Corollary}
\newtheorem{remark}{Remark}

\usepackage{hyperref}
\hypersetup{
    colorlinks=true,
    linkcolor=blue,
    citecolor=blue,
    urlcolor=blue
}
\usepackage{hypcap}

\newcommand{\E}{\mathbb{E}}
\newcommand{\V}{\mathbb{V}}

\newcommand{\threesum}{\sum_{i,j,k}}
\newcommand{\threesumprime}{\sum_{i',j',k'}}

\newcommand{\nd}{N}
\newcommand{\ndsquared}{N^2}

\newcommand{\ii}{\mathrm{i}}

\newcommand{\Cov}{\operatorname{Cov}}

\newcommand{\plotsize}{0.95}

\newcommand{\crosssum}{\hspace{-12pt} \sum_{\substack{i,k=1 \\ 1\le k-i \le r_{\max}}}^{N_\mathrm{site}} \hspace{-12pt}}
\usepackage{soul}
\usepackage[loadshadowlibrary, textsize=small, textwidth=3.5cm]{todonotes}

\newcommand{\revPG}[1]{%
  \begingroup
  \ifmmode
    \colorbox{RoyalBlue!20}{$\displaystyle #1$}%
  \else
    \sethlcolor{RoyalBlue!20}%
    \hl{#1}%
  \fi
  \endgroup
}

\begin{document}

\title{Scalable Lindblad Noise Learning via Stochastic Tensor-Network Simulation}
\date{April 2026}

\author{Alejandro R. Ramos Ramos\,\orcidlink{0009-0002-0922-1996}}
\affiliation{Zuse Institute Berlin}

\author{Maximilian Fröhlich\,\orcidlink{0009-0007-5276-2858}}
\affiliation{Weierstrass Institute}

\author{Aaron Sander\,\orcidlink{0009-0007-9166-6113}}
\affiliation{Technical University of Munich}

\author{Robert Wille\,\orcidlink{0000-0002-4993-7860}}
\affiliation{Technical University of Munich}
\affiliation{MQSC}
\affiliation{Software Competence Center Hagenberg (SCCH)}

\author{Martin Eigel\,\orcidlink{0000-0003-2687-4497}}
\affiliation{Weierstrass Institute}

\author{Patrick Gelß\,\orcidlink{0000-0002-3645-9513}}
\affiliation{Zuse Institute Berlin}

\author{Sebastian Pokutta}
\affiliation{Zuse Institute Berlin}

\begin{abstract}
    Learning dissipation rates in large-scale open quantum systems is a major obstacle for near-term quantum technologies, as existing Lindblad estimation methods are typically limited to small system sizes due to the computational complexity of repeatedly solving the Lindblad equation during optimization.
    Here, we propose a scalable noise-learning framework for Lindblad dissipation rates that combines a stochastic simulation method, the Tensor Jump Method (TJM), with gradient-free optimization of a least-squares cost-function defined on time series of local-observable expectation values.
    We demonstrate the approach on two noise models in the Ising model: a site-resolved (local) model, in which independent dissipation rates are learned for each site up to $N_{\mathrm{site}}=16$, and a spatially homogeneous (global) model with only seven parameters, scaled to $N_{\mathrm{site}}=160$ sites.
    We complement these numerical results with a series of exact, provable guarantees: the Frobenius variance of the TJM density-matrix estimator is shown to equal $(1-\mathrm{Tr}[\rho^2])/N_{\mathrm{traj}}$, an exact purity-based characterization of the stochastic estimation error; the corresponding purity evolution is proven to be monotonically non-increasing for Hermitian jump operators; and, under a finite covariance distance assumption, the standard deviation of the cost-function is shown to decrease with system size, so that fewer trajectories are needed to reach a fixed target accuracy as the system grows.
    Together, this combination of scalable numerics and rigorous theoretical guarantees positions TJM-based noise learning as a practical foundation for characterizing dissipation in large quantum devices and for guiding future work on error mitigation and quantum error correction.
\end{abstract}

\maketitle

\section{Introduction}

With the emergence of quantum technologies, understanding and controlling noise in quantum computers has become a central challenge.
In current devices, noise typically scales with the number of qubits and the depth of the quantum circuit, causing errors to accumulate and reduce the fidelity of computations \cite{georgopoulos_2021_modeling}.
Characterizing and modeling the dynamics of the quantum system, including its noise, also known as Lindbladian learning, is therefore essential to assess the performance of quantum hardware and to guide the development of effective error mitigation and quantum error correction strategies \cite{aseguinolaza_2024_error, vandenberg_2023_probabilistic, filippov_2023_scalable}.
A common modeling framework is provided by Lindblad master equations \cite{weimer_2021_simulation}, where unknown Hamiltonian and dissipative parameters must be inferred from experimentally accessible data such as measurement frequencies or time series of local observables.

One family of approaches to Lindbladian learning is based on the Ehrenfest theorem, where a set of observables is measured at different points in time and an objective function is formulated so that the Ehrenfest equation is satisfied as closely as possible, typically reducing the learning problem to a system of equations in the unknown Lindblad parameters.
This method was proposed in \cite{zubida_2021_optimal} for Hamiltonian learning and then extended to the dissipative case in \cite{stilckfranca_2024_efficient, olsacher_2025_hamiltonian}.
The most scalable realization of this idea is due to \textit{Van den Berg et al.}~\cite{berg_2025_largescale}, who learn a Lindbladian on a 156-qubit superconducting processor, the largest hardware demonstration of Lindbladian learning to date.
Their protocol obtains the gradients required by the Ehrenfest equation by fitting measured time series to sums of exponentially damped sinusoids and differentiating the fitted curves, so no explicit simulation of the Lindblad equation is required and the resulting least-squares problem is globally convex.
For scalability, however, the dissipator is restricted by construction to single-site Pauli rates (coefficients coupling two Pauli operators are set to zero unless the operators share support), while only the Hamiltonian may contain genuinely two-local terms; consequently, correlated multi-site noise of the form $L^{ZZ}_{ij} = Z_i \otimes Z_j$ cannot be learned within this framework.
Furthermore, the damped-sinusoid curve fitting requires sufficiently high time resolution to resolve fast coherent oscillations, and can suffer from aliasing otherwise.
A related, ansatz-free protocol with rigorous sample-complexity guarantees for sparse Lindbladians was recently proposed in \cite{ivashkov_2026_ansatzfree}, though it has not yet been demonstrated numerically at scale.

A second family of methods augments the physical Lindbladian model with a neural-network correction to make the resulting non-convex learning landscape more tractable.
\textit{Heightman et al.}~\cite{heightman_2026_lindbladian} tackle Lindbladian learning via maximum-likelihood estimation on Pauli measurements taken at multiple transient times, exploiting the fact that, unlike steady-state data, transient dynamics remain sensitive to the coherent part of the Lindbladian.
A neural differential equation (NDE) correction term is added to the physical model to help navigate the resulting non-convex likelihood landscape, and is then progressively removed during training via a curriculum so that the final model is a genuine, interpretable Lindbladian.
Because this approach requires repeated exact density-matrix simulation of the candidate dynamics, it is currently demonstrated only for systems up to $N=6$ qubits.

A third family of methods, which we refer to as simulation-assisted learning \cite{wang_2024_simulationassisted}, directly simulates the dynamics under a parameterized Lindbladian and fits its parameters by minimizing a discrepancy between predicted and measured data, including maximum-likelihood and related likelihood-based objectives \cite{samach_2022_lindblad,dobrynin_2024_compressedsensing}, divergences between measurement distributions \cite{benav_2020_direct,georgopoulos_2021_modeling}, or mean-squared-deviation fits to time series of local observables \cite{howard_2006_quantum,wang_2024_simulationassisted}.
\textit{Wang and Li}~\cite{wang_2024_simulationassisted} formulate the same trajectory-based least-squares problem as our work, and show that a semi-implicit Euler simulation of the Lindblad equation in Kraus form provides analytic gradients, enabling Levenberg--Marquardt optimization with quadratic convergence near the minimum; because the underlying simulation requires the full density matrix, typical system sizes studied with this approach range from 1 to 6 qubits \cite{benav_2020_direct, wang_2024_simulationassisted}.
\textit{Mangini et al.}~\cite{mangini_2024_tensor} instead characterize the noise channel accompanying a single circuit layer directly from tomographic measurement data, representing it as a locally-purified tensor-network density operator (LPDO) that is trained by gradient-based minimization (automatic differentiation) of a Kullback--Leibler-type loss; because the noise channel of one shallow layer only builds short-range correlations, this representation remains efficient at low bond dimension, allowing this approach to scale up to 20 qubits.
However, extending it to deep circuits or continuous-time Lindbladian dynamics would require the bond dimension to grow with the correlation range, and, unlike Wang and Li's approach or ours, it targets a single fixed channel rather than the generator of the full time evolution.

To overcome the scalability issues shared by all of the above simulation-based approaches, which are ultimately limited by the cost of repeatedly simulating exact density-matrix or state-vector dynamics, we propose a noise-learning framework that combines the Tensor Jump Method (TJM) \cite{sander_2025_largescale} for solving the Lindblad equation with the gradient-free optimization algorithms Covariance Matrix Adaptation Evolution Strategy (CMA-ES) \cite{hansen_2006_cma} and Bayesian Optimization (BO) \cite{frazier_2018_tutorial}, to optimize a mean-squared-deviation cost-function defined on time series of local observables.
We choose time series of local observables because they encode the dynamical information of the system and therefore contain signatures of the dissipative processes governed by the Lindblad equation.
Moreover, measuring local observables is significantly cheaper experimentally than performing full quantum state tomography, which scales exponentially with system size.

The scalability of the simulation is enabled by the TJM, a tensor-network-based quantum trajectory approach that efficiently simulates open quantum many-body dynamics by combining stochastic quantum jumps with matrix product state representations.
This method has been shown to scale to very large system sizes, allowing the simulation of spin models such as the Heisenberg chain with up to $10^3$ sites \cite{sander_2025_largescale}.
This scalability makes TJM particularly suitable for simulation-based learning of dissipative dynamics in regimes that are otherwise inaccessible to exact density-matrix approaches. Unlike Van den Berg et al.'s Ehrenfest-based approach, our cost-function places no structural restriction on the jump operators, so non-local or multi-site correlated noise terms enter simply as additional optimization parameters, a capability we demonstrate explicitly by learning ZZ correlated dephasing on all pairs within radius $r_{\max}=4$.
However, due to the stochastic nature of the TJM, the simulated trajectories introduce statistical fluctuations in the computed observables, which makes the objective function inherently noisy and requires careful treatment during optimization.
We further complement these numerical contributions with an exact Frobenius-variance theorem (Sec.~\ref{sec:exact_frob_var_tjm}) relating the stochastic estimation error of the TJM density-matrix estimator directly to the state purity, a standalone theoretical result with no counterpart in the works discussed above.

As a proof-of-principle application, we use our method to perform Lindbladian learning in the Ising model.
The Ising model serves as a natural benchmark for several reasons: it is one of the most fundamental and well-studied models in many-body physics, it captures essential features of interacting spin systems, and it is widely used to describe quantum simulators and near-term quantum devices.
Furthermore, its relatively simple structure, combined with nontrivial many-body dynamics, makes it an ideal testbed for assessing the ability of our approach to recover dissipative noise parameters from local dynamical observables in large-scale quantum systems.

The rest of this paper is organized as follows.
In Sec. \ref{sec:model_sys}, we introduce the model system by defining the Ising model and providing the necessary details of the Tensor Jump Method used to simulate the open-system dynamics.
In Sec. \ref{sec:opt_frame}, we present the optimization framework, where the least-squares cost-function is defined and additional details of the CMA-ES optimization procedure are discussed.
In Sec. \ref{sec:uniqueness}, we analyze the identifiability of the noise parameters and derive sufficient conditions under which a parameter cannot be recovered from the chosen observables.
In Sec. \ref{sec:rel_error}, we analyze the statistical fluctuations of the cost-function and derive an upper bound for its standard deviation in terms of the number of stochastic trajectories and system size.
In Sec. \ref{sec:results}, we present the results of our Lindbladian learning approach for both local and global noise models.
Finally, in Sec. \ref{sec:conclusions}, we summarize our findings and discuss possible directions for future work.


\section{Model System}\label{sec:model_sys}
Having motivated our scalable noise-learning approach and outlined the structure of the paper, we now turn to the concrete physical setting used throughout this work: an Ising spin chain subject to local and correlated dissipation, whose dynamics we describe using the Lindblad master equation:
\begin{multline}\label{eq:lindblad}
   \frac{d}{dt}\rho(t) = -\frac{\ii}{\hbar} [H_0,\rho] \\
   + \sum_{i=1}^{N_\mathrm{site}} \sum_{j=1}^{N_\mathrm{jump}} \gamma_i^{(j)} \!\left( L_i^{(j)} \rho L_i^{(j)\dag} - \frac{1}{2} \{ L_i^{(j)\dag} L_i^{(j)}, \rho \} \right) \\
   + \crosssum \gamma_{ik}^{ZZ} \!\left( L^{ZZ}_{ik} \rho L^{ZZ\,\dag}_{ik} - \frac{1}{2} \{ L^{ZZ\,\dag}_{ik} L^{ZZ}_{ik}, \rho \} \right)\!,
\end{multline}
where $\rho$ is the density matrix describing the state of the system, $H_0$ is the system Hamiltonian, $L_i^{(j)}$ are single-site jump operators with corresponding rates $\gamma_i^{(j)}$, and $L^{ZZ}_{ik}$ are two-site correlated dephasing (crosstalk) operators with rates $\gamma_{ik}^{ZZ}$.
Furthermore, $\ii$ denotes the imaginary unit, $\hbar$ is the reduced Planck constant, $N_{\mathrm{site}}$ is the number of sites, and $N_{\mathrm{jump}}$ is the number of jump operators per site.
In our case, for the single-site jump operators, we use the Pauli matrices $\sigma_j=(X,Y,Z)$ on each site, so $N_{\mathrm{jump}} = 3$.
This choice is motivated by the fact that Pauli twirling -- routinely applied on real hardware via randomized compiling -- tailors the physical noise experienced by a circuit into an effective Pauli channel, making the Pauli-Lindblad noise model considered here directly applicable to twirled devices \cite{geller_2013_efficient, berg_2024_techniques}.
Restricting the jump operators to the Pauli basis therefore targets a noise model of direct practical relevance, while keeping the parameter space minimal and directly interpretable.
The single-site jump operator $L_i^{(j)}$ is the $j$-th Pauli applied to site $i$:
\begin{equation*}
    L_i^{(j)} = I^{\otimes (i-1)} \otimes \sigma_j \otimes I^{\otimes (N_{\mathrm{site}}-i)} ~.
\end{equation*}
In addition to local noise, correlated ZZ dephasing (crosstalk) between nearby qubits is a leading error source in superconducting platforms \cite{sarovar_2020_detecting}.
We model it via two-site jump operators
\begin{equation*}
    L^{ZZ}_{ik} = I^{\otimes(i-1)} \otimes Z \otimes I^{\otimes(k-i-1)} \otimes Z \otimes I^{\otimes(N_{\mathrm{site}}-k)},
\end{equation*}
acting on sites $i$ and $k$.
Each pair $(i,k)$ with $1\le k-i\le r_{\max}$ (where $r_{\max}=4$) carries an independent rate $\gamma_{ik}^{ZZ}$.
The dissipation rates $\gamma_i^{(j)}$ and $\gamma_{ik}^{ZZ}$ together define the noise model and are the parameters we aim to learn in Equation \eqref{eq:lindblad}.

The Hamiltonian of the Ising model is given by:
\begin{equation*}
H_0 = -K \sum_{i=1}^{N_{\mathrm{site}}-1} Z_i Z_{i+1} - g \sum_{i=1}^{N_{\mathrm{site}}} X_i ~, 
\end{equation*}
$X_i$ and $Z_i$ are the corresponding Pauli matrices applied to the $i$-th site, for example,
\begin{equation*}
    X_i = I^{\otimes (i-1)} \otimes X \otimes I^{\otimes (N_{\mathrm{site}}-i)}~.
\end{equation*}
Here, $K$ sets the strength of the nearest-neighbor interaction, while $g$ is the strength of the transverse field.
In the following, we fix $K = 1$ and $g = 1$.
In physical units, both parameters have dimensions of energy, and time is naturally measured in units of $\hbar/K$.
Using the convention $\hbar = 1$ and measuring energies in units of $K$, the quantities $t$, $T$, and $dt$ are dimensionless.

As previously mentioned, to solve Equation \eqref{eq:lindblad} we use TJM, which is a tensor-network version of the Monte Carlo wave-function method \cite{dalibard_1992_wavefunction}.
It evolves Lindbladian dynamics by unraveling the mixed-state density matrix into trajectories of pure states $(\ket{\psi_p(t)},~ p=1,..., N_{\mathrm{traj}})$, represented as matrix product states (MPS), where each trajectory is governed by a non-Hermitian effective Hamiltonian
\begin{equation}
\begin{split}
    H_{\mathrm{eff}} = H_0 &- \frac{\ii}{2} \sum_{i=1}^{N_\mathrm{site}} \sum_{j=1}^{N_\mathrm{jump}} \gamma_i^{(j)} L_i^{(j)\dagger} L_i^{(j)} \\
    &~~~~~~~~~~~~~~ - \frac{\ii}{2} \crosssum \gamma_{ik}^{ZZ} L^{ZZ\,\dagger}_{ik} L^{ZZ}_{ik},
\end{split}
\end{equation}
by combining the time-dependent variational principle \cite{haegeman_2016_unifying} with a dissipative MPO-MPS contraction.
It then inserts dissipative jumps between time steps, sampled according to a probability distribution.
This distribution is computed from the norm loss after the non-Hermitian time evolution and from the noise parameters $\gamma_i^{(j)}$ and $\gamma_{ik}^{ZZ}$ in the Lindblad equation.
To approximate the expectation value of an observable $O$ at time $t$, we average over all trajectories:
\begin{multline}\label{eq:exp_value}
    \expval{O}(t) =  \mathrm{Tr}[\hat\rho_{N_{\mathrm{traj}}}(t) O] \\
    =   \frac{1}{N_{\mathrm{traj}}} \sum_{p=1}^{N_{\mathrm{traj}}} \mel{\psi_p(t)}{O}{\psi_p(t)}~.
\end{multline}
The approximation has the standard Monte Carlo error $\mathcal{O} \left(\frac{1}{\sqrt{N_{\mathrm{traj}}}}\right)$, a global time-step error of $\mathcal{O}(dt^2)$ arising from the second-order Strang splitting and dynamic TDVP integrator used in the TJM~\cite{sander_2025_largescale}, and a projection error due to the capped bond dimensions of the MPS.
Nevertheless, TJM shows high accuracy even for several hundred sites, as shown in \cite{sander_2025_largescale}.


\section{Optimization Framework}\label{sec:opt_frame}
We formulate the Lindbladian learning problem as finding the set of parameters $\boldsymbol{\gamma}$ that minimizes the mean-squared deviation between a set of reference and model-generated time-dependent expectation values; in other words, we minimize the function:
\begin{multline}\label{eq:cost_2}
    J(\boldsymbol{\gamma})= \frac{1}{N_{\mathrm{site}} N_{\mathrm{ob}} N_{\mathrm{time}}} \sum_{i=1}^{N_{\mathrm{site}}} \sum_{j=1}^{N_{\mathrm{ob}}}  \sum_{k=1}^{N_{\mathrm{time}}} \\
    \left[  \langle O_{ij} \rangle^{(\boldsymbol{\gamma})}(t_k)  -   \langle O_{ij} \rangle^{\mathrm{(ref)}}(t_k) \right]^2~,
\end{multline}
where $\langle O_{ij} \rangle^{(\boldsymbol\gamma)}$ denotes the observable expectation values generated from the noise model characterized by $\gamma_i^{(j)}$; $\langle O_{ij} \rangle^{\mathrm{(ref)}}$ denotes the reference expectation values; $N_{\mathrm{time}}$ is the number of time-discretization points; $N_{\mathrm{ob}}$ is the number of single-site observable operators per site; and $O_{ij}$ is the $j$-th single-site observable operator applied to the $i$-th site.
The initial time-discretization point in the trajectory is $t_1$, and the last one is $t_{N_{\mathrm{time}}}$.
As single-site observable operators, we use once again the Pauli matrices $(N_{\mathrm{ob}}=N_{\mathrm{jump}}=3)$:
\begin{equation*}
    O_{ij} = I^{\otimes (i-1)} \otimes \sigma_j \otimes I^{\otimes (N_{\mathrm{site}}-i)} ~.
\end{equation*}

The reference expectation values in Equation \eqref{eq:cost_2}, $\langle O_{ij} \rangle^{\mathrm{(ref)}}(t_k)$, could be obtained from experimental measurements, but in our case we generate them from reference dissipation rates $\gamma_i^{(j)\,\mathrm{(ref)}}$ and $\gamma_{ik}^{ZZ\,\mathrm{(ref)}}$, which we then try to learn via optimization.
Unless otherwise stated, these reference time-dependent expectation values are obtained by solving Equation \eqref{eq:lindblad} using TJM with $N_{\mathrm{traj}}=4000$, which guarantees sufficient precision.
As the initial density matrix, we choose all sites in the zero state:
\begin{equation*}
\rho(0)=\bigotimes_{j=1}^{N_{\mathrm{site}}}\ketbra{0}{0}_j~.
\end{equation*}
Furthermore, we simulate the system for a total time of $T=6$ ($t_1=0,~ t_{N_{\mathrm{time}}}=6$) with a time step of $dt=0.1$, for a total of $N_{\mathrm{time}}=61$ time-discretization points.
This choice of total time offers a good balance between simulation duration and feature capture in the time-dependent reference expectation values.
The MPS bond dimension is capped at $\chi_{\mathrm{max}} = 8$ throughout all simulations.
Such a modest cap is justified by the dissipative nature of the dynamics, which interrupts the coherent buildup of entanglement: the original TJM paper~\cite{sander_2025_largescale} shows that $\chi=4$ already reproduces the magnetization of a noisy 30-site Heisenberg chain against a numerically-exact MPO solver ($D=400$), and uses $\chi=4$ throughout their $1000$-site simulations. Our choice of $\chi_{\mathrm{max}}=8$ is therefore conservative in comparison.

Due to the stochastic nature of our simulation process, the cost-function in Equation \eqref{eq:cost_2} is noisy (each evaluation almost surely returns slightly different values of $J(\boldsymbol{\gamma})$ for the same $\boldsymbol{\gamma}$, as analyzed in Section~\ref{sec:rel_error}).
Furthermore, $J(\boldsymbol{\gamma})$ is expensive to evaluate because it requires numerically solving the Lindblad equation, and we do not have access to its gradients.
Finite-difference gradients are also time-consuming and of limited use in noisy settings.
As discussed previously, one algorithm we use to optimize $J(\boldsymbol{\gamma})$ is CMA-ES \cite{hansen_2006_cma}, which maintains a multivariate Gaussian search distribution over the parameters, repeatedly samples candidate solutions, and selects the best ones to update the center of the search distribution.
It adapts the covariance matrix (and step size) based on successful search steps, so the sampling distribution learns correlations and stretches and rotates toward promising directions in the landscape.
The other optimization method we use is BO \cite{frazier_2018_tutorial}, which builds a probabilistic surrogate model (typically a Gaussian process) of the objective function and uses it to guide the search.
At each iteration, it selects new candidate points by maximizing an acquisition function that balances exploration (sampling uncertain regions) and exploitation (refining near promising values).
The surrogate is then updated with the new observations, allowing the method to efficiently locate optima with relatively few expensive function evaluations.


\section{Identifiability Analysis}\label{sec:uniqueness}

Characterizing the complete landscape of $J(\boldsymbol{\gamma})$, including general conditions for the existence and uniqueness of a global minimizer, is a challenging problem that depends in a complex way on the noise model, the system dynamics, and the choice of observables.
However, sufficient conditions can be given for the cost-function to carry \emph{no information} about a particular parameter $\gamma_l$:
\begin{equation*}
    \frac{\partial J(\boldsymbol{\gamma})}{\partial \gamma_l} = 0, \quad \forall \boldsymbol{\gamma}~.
\end{equation*}
This condition implies that the cost-function is flat along the $\gamma_l$ direction, and the minimizer is non-unique.
We call $\gamma_l$ \emph{unidentifiable} in this case. 
Expanding the partial derivative of the cost-function with respect to $\gamma_l$ and absorbing the positive prefactor into the zero on the right-hand side gives:
\begin{multline}\label{eq:partial_J_obs}
    \sum_{n=1}^{N_{\mathrm{OBS}}} \sum_{k=1}^{N_{\mathrm{time}}} \left[  \langle O_{n} \rangle^{(\boldsymbol{\gamma})}(t_k)  -   \langle O_{n} \rangle^{\mathrm{(ref)}}(t_k) \right] \\
    \times \frac{\partial }{\partial \gamma_l} \langle O_{n} \rangle^{(\boldsymbol{\gamma})}(t_k) = 0, \quad \forall \boldsymbol{\gamma}~.
\end{multline}
To reduce the number of indices we have collapsed the double summation over single-site observables and sites in Equation~\eqref{eq:cost_2} into a single summation over the total number of observables $N_{\mathrm{OBS}}=N_{\mathrm{site}} \times N_{\mathrm{ob}}$.
Since Equation~\eqref{eq:partial_J_obs} is a weighted sum over observables and time steps of terms proportional to $\partial \langle O_n \rangle^{(\boldsymbol{\gamma})} / \partial \gamma_l$, a sufficient condition for non-identifiability is that
\begin{equation}\label{eq:partial_O}
    \frac{\partial}{\partial \gamma_l} \langle O_n \rangle^{(\boldsymbol{\gamma})}(t_k)= \mathrm{Tr} \left[ \frac{\partial \rho (t_k)}{\partial \gamma_l} O_n\right] = 0, ~~~\forall n, k, \boldsymbol{\gamma}.
\end{equation}
In other words, we are requiring that the time-dependent expectation values of the observables $\{O_n\}$ are not affected by the parameter $\gamma_l$.
We can obtain an expression for $\frac{\partial \rho (t_k)}{\partial \gamma_l}$ by differentiating the Lindblad equation~\eqref{eq:lindblad} with respect to $\gamma_l$, interchanging the time and $\gamma_l$ derivatives and then integrating from $0$ to $t_k$:
\begin{multline}\label{eq:drho_dgamma}
     \frac{\partial \rho(t_k)}{\partial \gamma_l} = \int_0^{t_k} \Bigg[ -\ii \left[H_0, \frac{\partial \rho(s)}{\partial \gamma_l}\right] \\
     + \sum_{m=1}^{N_{\mathrm{JUMPS}}} \gamma_m \Bigg( L_m \frac{\partial \rho(s)}{\partial \gamma_l} L_m^\dagger
     - \frac{1}{2} \left\{ L_m^\dagger L_m, \frac{\partial \rho(s)}{\partial \gamma_l}\right\} \Bigg) \\
    + L_l \rho(s) L_l^\dagger - \frac{1}{2} \left\{ L_l^\dagger L_l, \rho(s) \right\} \Bigg] ds ,
\end{multline}
with initial condition $\partial \rho(0) / \partial \gamma_l = 0$, since $\rho(0)$ does not depend on $\boldsymbol{\gamma}$.
Here we have collapsed the double summation from Equation~\eqref{eq:lindblad} into one summation over the total number of jump operators $N_{\mathrm{JUMPS}} = N_{\mathrm{site}} \times N_{\mathrm{jump}} + N_{\mathrm{ZZ}}$, where $N_{\mathrm{ZZ}} = \sum_{r=1}^{r_{\max}}(N_{\mathrm{site}}-r)$ is the number of ZZ crosstalk operator pairs.
Substituting~\eqref{eq:drho_dgamma} into~\eqref{eq:partial_O} and using cyclicity of the trace, we obtain
\begin{multline}\label{eq:partial_O_int}
    \frac{\partial}{\partial \gamma_l} \langle O_n \rangle^{(\boldsymbol{\gamma})}(t_k) = \int_0^{t_k} \mathrm{Tr}\!\Bigg[ \frac{\partial \rho(s)}{\partial \gamma_l} \\
    \times \left( -\ii [O_n, H_0] + \sum_{m=1}^{N_{\mathrm{JUMPS}}} \gamma_m A_{mn} \right) \\
    + \rho(s)\, A_{ln} \Bigg] ds~,
\end{multline}
where the operator $A_{mn}$ is defined as
\begin{equation}\label{eq:Amn}
    A_{mn} = L_m^\dagger O_n L_m - \frac{1}{2} \left\{ L_m^\dagger L_m,\, O_n \right\}.
\end{equation}
We can now state sufficient conditions under which $\gamma_l$ is unidentifiable.

\begin{proposition}[Sufficient conditions for non-identifiability]\label{prop:non_id}
The parameter $\gamma_l$ is unidentifiable from the set of observables $\{O_n\}$, i.e.,
$\partial \langle O_n \rangle^{(\boldsymbol{\gamma})}(t_k) / \partial \gamma_l = 0$
for all $n$, $t_k$ and $\boldsymbol{\gamma}$,
if either of the following conditions holds:
\begin{enumerate}
    \item \textbf{State insensitivity:} $\partial \rho(t_k) / \partial \gamma_l = 0 \quad \forall t_k, \boldsymbol{\gamma}$.
    \item \textbf{Observable commutativity:} $[O_n, H_0] = 0$ and $A_{mn} = 0 \quad \forall m$.
\end{enumerate}
\end{proposition}

\begin{proof}
Condition~1 makes the expression \eqref{eq:partial_O} vanish directly.
Under condition~2, $A_{mn} = 0$,  in particular, $A_{ln} = 0$, and $[O_n, H_0] = 0$ by assumption.
Hence the integrand in~\eqref{eq:partial_O_int} vanishes.
\end{proof}

\begin{remark}
Condition~1 requires the dissipator of $L_l$ to vanish along the \emph{entire} trajectory, i.e.,
$$L_l \rho(t) L_l^\dagger - \frac{1}{2} \left\{ L_l^\dagger L_l, \rho(t) \right\} = 0 \quad \forall\, t \geq 0,~ \forall\, \boldsymbol{\gamma}~.$$
When this holds, the source term in the ODE for $\partial\rho/\partial\gamma_l$ is identically zero, so $\partial\rho(t)/\partial\gamma_l = 0$ for all $t$ given the zero initial condition $\partial\rho(0)/\partial\gamma_l = 0$.
\begin{enumerate}
    \item A sufficient case is when $\rho(t) = \mathbb{I}/d$ for all $t$ and $L_l$ is a normal operator $\left(\text{i.e.  }[L_l^\dagger, L_l]=0\right)$. The condition on $\rho(t)$ is satisfied if $\rho(0) = \mathbb{I}/d$ and the Lindbladian is unital, $\sum_{m=1}^{N_{\mathrm{JUMPS}}} \gamma_m [L_m,L_m^\dagger]=0$. Since non-identifiability is required for all $\boldsymbol{\gamma}$, and the rates $\gamma_m$ are independent, this must hold for every rate vector, which is equivalent to demanding that each jump operator be normal, $[L_m, L_m^\dagger]=0 \;\forall m$.
    \item Another sufficient case is when $\rho(0)=|\psi_0\rangle\langle\psi_0|:~L_l|\psi_0\rangle = \lambda |\psi_0\rangle$ $($i.e. $|\psi_0\rangle$ is an eigenstate of $L_l$$)$, $[L_l,L_m]=0$ for all $m$, $L_m$ is normal for all $m$, and $[L_l,H_0]=0$.
\end{enumerate}
\end{remark}

\begin{remark}
If we look at the Lindblad--Ehrenfest equation, obtained by differentiating $\langle O_n\rangle = \mathrm{Tr}[\rho\, O_n]$ with respect to time,
\begin{equation}\label{eq:ehrenfest}
    \frac{d}{dt}\langle O_n \rangle = -\ii\langle [O_n, H_0]\rangle + \sum_{m=1}^{N_{\mathrm{JUMPS}}} \gamma_m \langle A_{mn} \rangle~,
\end{equation}
we can see that Condition~2 of Proposition~\ref{prop:non_id} also implies that $\langle O_n \rangle$ remains constant in time for all $\boldsymbol{\gamma}$.
$A_{mn}=0$ is satisfied for example if observable $O_n$ commutes with $L_m$ $([L_m, O_n] = 0)$, which happens when the operators $L_m$ and $O_n$ act on different sites, or when they act on the same site and commute.
\end{remark}

\begin{figure}[t]
    \centering
    \includegraphics[width=\plotsize\linewidth]{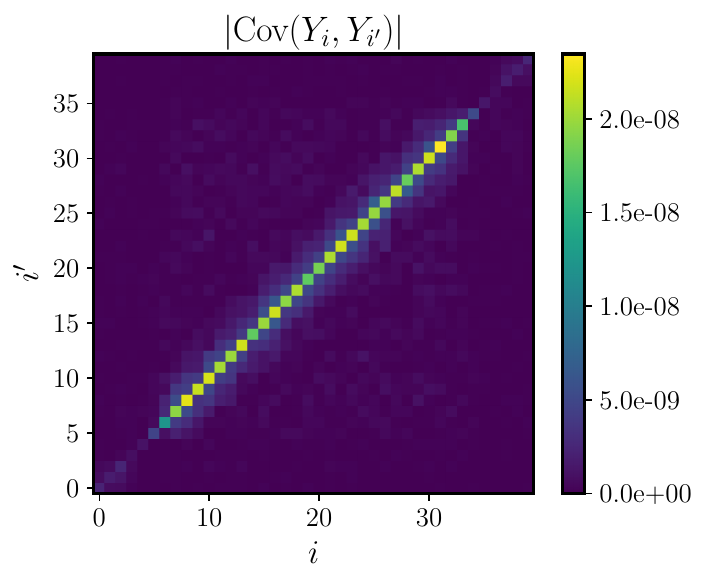}
    \caption{Empirical covariance matrix $|\mathrm{Cov}(Y_{ijk},Y_{i'j'k'})|$ for the $X$ observable ($j=j'=1$) at time steps $k=37$ and $k'=41$, estimated from $N_{\mathrm{traj}}=1000$ trajectories for a system size of $N_{\mathrm{site}}=40$. These specific values of $j$, $j'$, $k$ and $k'$ were selected because, among all time-step pairs, they exhibit the largest maximum covariance distance $\ell$, of about $4$.}
    \label{fig:cov_matrix_ntraj_1000}
\end{figure}

The sufficient conditions for non-identifiability stated in Proposition~\ref{prop:non_id} are intended as a practical rule of thumb for identifying experimental setups to avoid.
In our setup (Pauli and ZZ crosstalk jump operators, Pauli observables, and the Ising Hamiltonian with a nonzero transverse field), none of these conditions are fulfilled: the observables do not commute with $H_0$ due to the transverse-field term, which rules out Condition~2; and although the Lindbladian is unital and the jump operators are normal, the initial state $\rho(0) = |0\rangle\langle 0|^{\otimes N_\mathrm{site}} \neq \mathbb{I}/d$ is not maximally mixed, so $\rho(t)$ does not remain at $\mathbb{I}/d$ and the sufficient case of Condition~1 does not apply.
It should be emphasized, however, that avoiding these sufficient conditions does not guarantee identifiability; they rule out specific failure modes but do not constitute a positive proof that $\gamma_l$ is recoverable from $\{O_n\}$.
The most direct test of identifiability is to scan $J(\boldsymbol{\gamma})$ along the $\gamma_l$ direction and verify that the cost-function has a well-defined minimum, but this becomes impractical when the number of parameters is large.


\section{Statistical Fluctuations of the TJM Estimator and Cost-Function}
 
\subsection{Exact Frobenius variance of the TJM estimator} \label{sec:exact_frob_var_tjm}
 
Before bounding the cost-function variance, we first characterize the
stochastic error of the TJM estimator $\hat\rho_{N_\mathrm{traj}}(t)$ itself,
independently of any optimization objective. This error admits an exact
characterization in terms of the purity of the Lindbladian solution, and
the relevant quantity can be estimated directly from the sampled
trajectories.

\begin{theorem}[Exact Frobenius variance]
\label{thm:exact_frobenius_variance}
Let $\rho(t)$ be the solution of the Lindblad master equation~\eqref{eq:lindblad} at time $t \in [0,T]$, and let
\begin{equation}
\hat\rho_{N_\mathrm{traj}}(t)
= \frac{1}{N_\mathrm{traj}} \sum_{i=1}^{N_\mathrm{traj}}
|\Psi_i(t)\rangle\langle\Psi_i(t)|
\label{eq:rhoN_estimator}
\end{equation}
be the TJM estimator of $\rho(t)$: a random variable whose randomness is
inherited from the $N_\mathrm{traj}$ independently and identically sampled
quantum trajectories $|\Psi_i(t)\rangle$, each satisfying
$\mathbb{E}[|\Psi_i(t)\rangle\langle\Psi_i(t)|] = \rho(t)$ in the
full-bond-dimension limit. Then for every $N_\mathrm{traj} \in \mathbb{N}$
and every $t \in [0,T]$, the variance of $\hat\rho_{N_\mathrm{traj}}(t)$
with respect to the Frobenius norm is given by
\begin{equation}
\mathbb{V}_F\!\left[\hat\rho_{N_\mathrm{traj}}(t)\right]
=
\frac{1 - \mathrm{Tr}[\rho(t)^2]}{N_\mathrm{traj}}.
\label{eq:exact_frobenius_variance}
\end{equation}
\end{theorem}

\begin{proof}
For ease of notation, we drop the time parameter $t$ and write
$X_i := |\Psi_i\rangle\langle\Psi_i|$ for $i = 1, \dots, N_\mathrm{traj}$.
The trajectories are sampled independently and identically distributed,
and the full-bond-dimension TJM is equivalent to the Lindblad master
equation, so $\mathbb{E}[X_i] = \rho$ for all $i$. Moreover,
$\hat\rho_{N_\mathrm{traj}} = \frac{1}{N_\mathrm{traj}} \sum_i X_i$ is their sample mean. By
independence,
\begin{align}
\mathbb{V}_F\!\left[\hat\rho_{N_\mathrm{traj}}\right]
&= \frac{1}{N_\mathrm{traj}^2}
\sum_{i=1}^{N_\mathrm{traj}} \mathbb{V}_F[X_i]
= \frac{1}{N_\mathrm{traj}}\,\mathbb{V}_F[X_1].
\label{eq:variance_reduction}
\end{align}
Using the Hermiticity of $X_1$ and $\rho$,
\begin{align}
\mathbb{V}_F[X_1]
&= \mathbb{E}\!\left[\|X_1-\rho\|_F^2\right]
= \mathbb{E}\!\left[\mathrm{Tr}[(X_1-\rho)^2 ]\right] \nonumber\\
&= \mathbb{E}\!\left[\mathrm{Tr}[X_1^2]\right]
- 2\,\mathrm{Tr}[\mathbb{E}[X_1]\rho] + \mathrm{Tr}[\rho^2].
\label{eq:single_variance_expand}
\end{align}
Since each $X_1 = |\Psi_1\rangle\langle\Psi_1|$ is a rank-one projector,
$X_1^2 = X_1$ and therefore $\mathrm{Tr}[X_1^2] = 1$. Combined with
$\mathbb{E}[X_1] = \rho$, this gives
\begin{equation}
\mathbb{V}_F[X_1] = 1 - 2\,\mathrm{Tr}[\rho^2] + \mathrm{Tr}[\rho^2]
= 1 - \mathrm{Tr}[\rho^2].
\end{equation}
Substituting into~\eqref{eq:variance_reduction} yields
Eq.~\eqref{eq:exact_frobenius_variance}.
\end{proof}

\noindent The corresponding standard deviation is
\begin{equation}
\sigma_F\!\left[\hat\rho_{N_\mathrm{traj}}(t)\right]
= \sqrt{\frac{1 - \mathrm{Tr}[\rho(t)^2]}{N_\mathrm{traj}}}.
\label{eq:exact_frobenius_std}
\end{equation}
The stochastic estimation error is thus fully determined by the purity
$P(t) := \mathrm{Tr}[\rho(t)^2]$ of the underlying Lindbladian solution.
Although $P(t)$ is in general unknown, it can be estimated without bias
directly from the sampled trajectories:

\begin{corollary}[Unbiased purity estimator]
\label{cor:purity_estimator}
For two independent trajectories $|\Psi_1(t)\rangle$ and
$|\Psi_2(t)\rangle$ sampled by the TJM,
$P(t) = \mathbb{E}\!\left[|\langle\Psi_1(t),\Psi_2(t)\rangle|^2\right]$.
In particular,
\begin{equation}
\widehat{P}_{N_\mathrm{traj}}(t)
:= \frac{1}{N_\mathrm{traj}(N_\mathrm{traj}-1)}
\sum_{\substack{i,j=1\\ i\neq j}}^{N_\mathrm{traj}}
\left|\langle\Psi_i(t),\Psi_j(t)\rangle\right|^2
\label{eq:purity_estimator}
\end{equation}
is an unbiased estimator of $P(t)$, and admits the equivalent form
\begin{equation}
\widehat{P}_{N_\mathrm{traj}}(t)
= \frac{N_\mathrm{traj}\,\mathrm{Tr}\!\left[\hat\rho_{N_\mathrm{traj}}(t)^2\right] - 1}
       {N_\mathrm{traj} - 1}.
\label{eq:purity_estimator_rhoN}
\end{equation}
\end{corollary}

\begin{proof}
For $i \neq j$, independence gives
$\mathbb{E}[\mathrm{Tr}[X_i X_j]] = \mathrm{Tr}[\mathbb{E}[X_i]\mathbb{E}[X_j]]
= \mathrm{Tr}[\rho^2]$. Since
$\mathrm{Tr}[X_i X_j] = |\langle\Psi_i,\Psi_j\rangle|^2$, averaging over
all pairs $i \neq j$ proves unbiasedness. For the second form, expand
\begin{align}
\mathrm{Tr}[\hat\rho_{N_\mathrm{traj}}^2]
&= \frac{1}{N_\mathrm{traj}^2}\sum_{i=1}^{N_\mathrm{traj}}
   \mathrm{Tr}[X_i^2]
 + \frac{1}{N_\mathrm{traj}^2}
   \sum_{\substack{i,j\\i\neq j}}\mathrm{Tr}[X_i X_j] \nonumber\\
&= \frac{1}{N_\mathrm{traj}}
 + \frac{N_\mathrm{traj}-1}{N_\mathrm{traj}}\,
   \widehat{P}_{N_\mathrm{traj}},
\end{align}
since $\mathrm{Tr}[X_i^2] = 1$. Solving for
$\widehat{P}_{N_\mathrm{traj}}$ yields
Eq.~\eqref{eq:purity_estimator_rhoN}.
\end{proof}

The combination of
Theorem~\ref{thm:exact_frobenius_variance} and
Corollary~\ref{cor:purity_estimator} yields a practical diagnostic: at
any time $t$, the stochastic estimation error of the TJM can be
quantified empirically from the same set of trajectories used in the
simulation, without recourse to an analytical model of $\rho(t)$.
 
The purity itself depends on the dissipative part of the dynamics, which
ties the simulation error directly to the noise parameters we aim to
learn. We summarize this dependence in the following two corollaries
(proofs in Appendix~\ref{appendix_lemmas}).

\begin{corollary}[Hermitian jump operators]
\label{cor:hermitian_jump}
If all jump operators in Eq.~\eqref{eq:lindblad} are Hermitian, i.e.
$L_m = L_m^\dagger$ for all $m$, then
\begin{equation}
\frac{d}{dt}P(t)
= -\sum_{m} \gamma_m \,
   \left\|\,[L_m,\rho(t)]\,\right\|_F^2
\;\le\; 0,
\label{eq:purity_monotone}
\end{equation}
i.e., the purity is monotonically non-increasing.
\end{corollary}

\noindent In our setup all jump operators (single-site Paulis and two-site ZZ
crosstalk operators) are Hermitian, so
Eq.~\eqref{eq:purity_monotone} applies directly and identifies the
commutator structure $\|[L_m, \rho]\|_F^2$ as the mechanism through
which each dissipation rate $\gamma_m$ controls the growth of the TJM
estimator error.

\begin{corollary}[Short-time behavior for pure initial states]
\label{cor:short_time}
If the initial state is pure, $\rho(0) = |\psi_0\rangle\langle\psi_0|$,
and $\sum_{m}\gamma_m(\Delta_{\psi_0}L_m)^2 > 0$, then
\begin{multline}
\sigma_F\!\left[\hat\rho_{N_\mathrm{traj}}(t)\right]
= \sqrt{\frac{2t}{N_\mathrm{traj}}
        \sum_{m=1}^{N_{\mathrm{JUMPS}}} \gamma_m\,(\Delta_{\psi_0} L_m)^2} \\
\,+\, O\!\left(\frac{t^{3/2}}{\sqrt{N_\mathrm{traj}}}\right),
\label{eq:short_time_scaling}
\end{multline}
where the sum runs over all $N_{\mathrm{JUMPS}}$ jump operators and
$(\Delta_{\psi_0} L_m)^2
:= \langle\psi_0|L_m^\dagger L_m|\psi_0\rangle
 - |\langle\psi_0|L_m|\psi_0\rangle|^2$
denotes the squared state fluctuation of $L_m$ in the state $|\psi_0\rangle$.
For Hermitian $L_m$, this coincides with the usual quantum-mechanical
variance/uncertainty squared. The hypothesis
$\sum_{m}\gamma_m(\Delta_{\psi_0}L_m)^2 > 0$ is necessary: being a sum of
non-negative terms, the fluctuation sum vanishes precisely when, for
every $m$, either $\gamma_m = 0$ or $|\psi_0\rangle$ is an eigenstate of
$L_m$, and the leading $\sqrt{t}$ term is then absent, leaving
$\sigma_F[\hat\rho_{N_\mathrm{traj}}(t)] = O(t/\sqrt{N_\mathrm{traj}})$.
In the closed-system limit $\gamma_m \equiv 0$ the dynamics is unitary, a
single trajectory is exact, and this residual vanishes identically,
$\sigma_F[\hat\rho_{N_\mathrm{traj}}(t)] = 0$.
\end{corollary}

\noindent The initial state used throughout our simulations is the pure
product state $\rho(0) = \bigotimes_j |0\rangle\langle 0|_j$, which is an
eigenstate of the diagonal ($Z$ and $ZZ$) jump operators but not of the
transverse $X$/$Y$ Pauli jumps; the fluctuation sum
$\sum_{m}\gamma_m(\Delta_{\psi_0}L_m)^2$ is therefore strictly positive
whenever any transverse jump rate is nonzero, and
Eq.~\eqref{eq:short_time_scaling} then provides an explicit short-time
scaling of the TJM estimator error in terms of the initial-state squared
fluctuations $(\Delta_{\psi_0}L_m)^2$, without the covariance-distance
assumption needed for the cost-function bound below.
 
Equation~\eqref{eq:exact_frobenius_variance} also identifies a regime in
which the TJM's sampling error vanishes: whenever the initial state lies
on a submanifold preserved by the Lindblad evolution (e.g.\ eigenstates
of dephasing operators), the purity remains $\mathrm{Tr}[\rho(t)^2] = 1$
for all $t$, and consequently
$\mathbb{V}_F[\hat\rho_{N_\mathrm{traj}}(t)] = 0$ independently of
$N_\mathrm{traj}$, though other error sources intrinsic to TJM (Trotter
splitting error and MPS bond-dimension truncation) remain unaffected.
This is the all-time strengthening of the degenerate case in
Corollary~\ref{cor:short_time}: there, the vanishing of
$\sum_{m}\gamma_m(\Delta_{\psi_0}L_m)^2$ removes only the leading
$\sqrt{t}$ growth and leaves an $O(t/\sqrt{N_\mathrm{traj}})$ residual,
whereas here the defining condition is preserved for all $t$ and the
error vanishes identically.
 
\bigskip
 
\noindent The Frobenius variance characterizes the TJM estimator
directly. The cost-function $J(\boldsymbol{\gamma})$ defined in
Eq.~\eqref{eq:cost_2}, however, aggregates squared deviations
over sites, observables, and times. We now bound its standard deviation
under an additional finite covariance distance assumption that is
empirically supported in our setup (Fig.~\ref{fig:cov_matrix_ntraj_1000}). 

\subsection{Estimation of the Standard Deviation of the Cost-Function} \label{sec:rel_error}

To simplify the notation, we introduce the shorthand
\begin{align*}
    X_{ijk} &:= \langle O_{ij} \rangle^{(\boldsymbol{\gamma})}(t_k)~, &
    \mu_{ijk} &:= \mathbb{E}[X_{ijk}]~,\\
    \sigma_{ijk} &:= \sigma(X_{ijk})~, &
    \mu_{ijk}^{\mathrm{(ref)}} &:= \langle O_{ij} \rangle^{\mathrm{(ref)}}(t_k)~,\\
    Y_{ijk} &:= \left( X_{ijk} - \mu_{ijk}^{\mathrm{(ref)}} \right)^2~, &
    N &:= N_{\mathrm{site}} N_{\mathrm{ob}} N_{\mathrm{time}}~,
\end{align*}
where $X_{ijk}$, $Y_{ijk}$, and $\hat{J}$ are random variables. The cost-function in Equation~\eqref{eq:cost_2} becomes the random variable
\begin{equation}\label{eq:cost_2_X}
    \hat{J} = \frac{1}{\nd} \sum_{i,j,k} Y_{ijk}~,
\end{equation}
with $i\in\{1,\ldots,N_{\mathrm{site}}\}$, $j\in\{1,\ldots,N_{\mathrm{ob}}\}$, and $k\in\{1,\ldots,N_{\mathrm{time}}\}$.

\begin{figure}[t]
    \centering
    \includegraphics[width=\plotsize\linewidth]{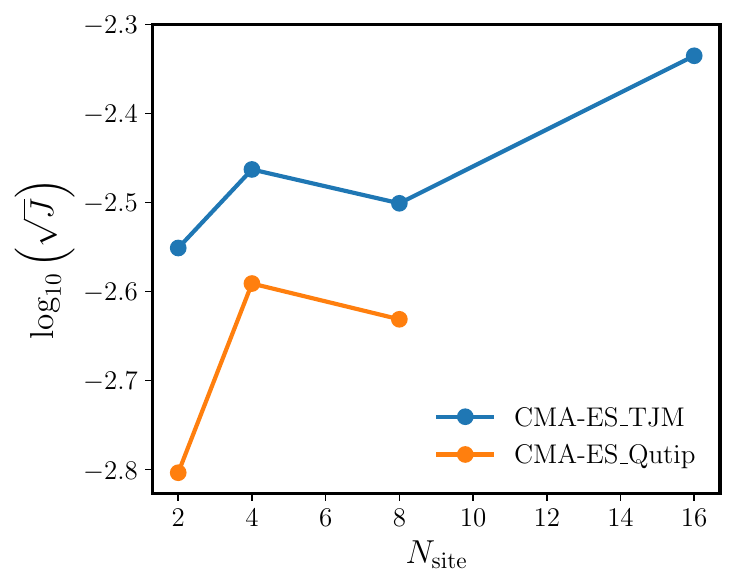}
    \caption{Optimal cost-function value vs. number of sites for a local-noise-model optimization. Both curves are obtained using the CMA-ES optimizer with TJM (blue) and QuTiP (orange) Lindblad solvers. We plot $\log_{10}\left(\sqrt{J}\right)$ rather than $\log_{10}(J)$ so that the plotted quantity has the same units as the observables.}
    \label{fig:loss_local_noise}
\end{figure}

We now derive an upper bound for the variance of the cost:
\begin{align*}
\V\left[\hat{J}\right] &= \V \left[ \frac{1}{\nd} \threesum Y_{ijk} \right] \\
 & = \frac{1}{\ndsquared} \threesum \threesumprime \Cov (Y_{ijk},Y_{i'j'k'})  \\
 & \leq \frac{1}{\ndsquared} \threesum \threesumprime \left| \Cov (Y_{ijk},Y_{i'j'k'}) \right|~.
\end{align*}
Unlike the exact result of
Sec.~\ref{sec:exact_frob_var_tjm}, the following bound requires an
additional structural assumption: covariances between sites separated
by more than a fixed maximum covariance distance $\ell$ decay exponentially. Specifically,
we assume that there exists $\ell \in \mathbb{N}$ and $0 < \epsilon < 1$
such that for $|i - i'| > \ell$, the covariance scales as
$\epsilon^{|i-i'|-\ell}$. We further assume that cross-time and
cross-observable covariances are dominated by the spatial decay; this
is consistent with the empirical covariance matrix shown in
Fig.~\ref{fig:cov_matrix_ntraj_1000}. Under this assumption, bounding the
near-diagonal entries with the Cauchy--Schwarz inequality and the far tail by the
assumed decay, the covariance matrix entries obey:
\begin{equation*}
\left| \Cov (Y_{ijk},Y_{i'j'k'}) \right| \leq
\begin{cases}
\sqrt{\V[Y_{ijk}]\,\V[Y_{i'j'k'}]}, & |i-i'| \le \ell, \\
\epsilon^{|i-i'|-\ell}, & |i-i'| > \ell,
\end{cases}
\end{equation*}
Applying Lemma~\ref{lem:var_y_bound} (all auxiliary results can be found in Appendix~\ref{appendix_lemmas}), which gives $\V[Y_{ijk}] \le c_{ijk}\,\sigma_{ijk}^2$ with
\begin{equation*}
c_{ijk} = \left( M_{ijk} + |\mu_{ijk}| + 2 \left| \mu_{ijk} - \mu^{\mathrm{(ref)}}_{ijk} \right|\right)^2
\end{equation*}
($M_{ijk}$ being the bound for $X_{ijk}$), together with the AM--GM inequality $\sqrt{ab}\le(a+b)/2$, the near-diagonal ($|i-i'|\le \ell$) entries obey
\begin{equation*}
\left| \Cov (Y_{ijk},Y_{i'j'k'}) \right| \leq \frac{c_{ijk}\,\sigma_{ijk}^2 + c_{i'j'k'}\,\sigma_{i'j'k'}^2}{2},
\end{equation*}
while the far tail ($|i-i'|>\ell$) keeps the bound $\epsilon^{|i-i'|-\ell}$. We sum the two
branches separately. The near-diagonal contribution is
$S_{\mathrm{band}} := \sum_{|i-i'|\le \ell}\big|\Cov(Y_{ijk},Y_{i'j'k'})\big|$, where the
site indices obey $|i-i'|\le \ell$ while $j,k,j',k'$ range freely over all
$N_{\mathrm{ob}}N_{\mathrm{time}}$ values on each side. Inserting the AM--GM bound
$|\Cov| \le \tfrac{1}{2}\big(c_{ijk}\sigma_{ijk}^2 + c_{i'j'k'}\sigma_{i'j'k'}^2\big)$
and using that the summation domain is symmetric under
$(i,j,k)\leftrightarrow(i',j',k')$, so that the two terms sum to the same value,
gives
\begin{align*}
S_{\mathrm{band}}
&\le \sum_{|i-i'|\le \ell} c_{ijk}\sigma_{ijk}^2 \\
&\le (2\ell+1)\,N_{\mathrm{ob}} N_{\mathrm{time}} \sum_{ijk} c_{ijk}\sigma_{ijk}^2,
\end{align*}
where the second line counts, for each fixed $(i,j,k)$, the
$\#\{i':|i-i'|\le \ell\}\le 2\ell+1$ band partners and the $N_{\mathrm{ob}}N_{\mathrm{time}}$
unrestricted $(j',k')$ pairs (the summand being independent of $i',j',k'$).
Combining this with the far-tail sum, evaluated using Lemma~\ref{lem:sum_eps}, and
dividing by $\ndsquared$,
\begin{align*}
\V[\hat{J}] &\leq \frac{(2\ell+1)\,N_{\mathrm{ob}} N_{\mathrm{time}}}{\ndsquared} \sum_{ijk} c_{ijk}\sigma_{ijk}^2 \\
    &\quad + \frac{N_{\mathrm{ob}}^2 N_{\mathrm{time}}^2}{\ndsquared}\,\frac{2\epsilon}{(1-\epsilon)^2} \Big( N_{\mathrm{site}} + \ell\epsilon \\
    &\qquad\qquad + \epsilon^{N_{\mathrm{site}}-\ell} - N_{\mathrm{site}}\epsilon - \ell - 1 \Big) \\
    &\leq \frac{1}{N_{\mathrm{site}}} \Bigg[
      (2\ell+1)\,\overline{c\sigma^2} \\
    &\qquad + \frac{2\epsilon}{(1-\epsilon)^2} \Bigl( 1 + \frac{\ell\epsilon}{N_{\mathrm{site}}}
      + \frac{\epsilon^{N_{\mathrm{site}}-\ell}}{N_{\mathrm{site}}} \Bigr) \Bigg] ~,
\end{align*}
where
\begin{equation*}
\overline{c\sigma^2} := \frac{1}{\nd}\sum_{ijk} c_{ijk}\,\sigma_{ijk}^2
\end{equation*}
is the average of $c_{ijk}\sigma_{ijk}^2$ over all sites, observables, and time steps.

\begin{figure*}[t]
    \centering
    \includegraphics[width=0.85\textwidth]{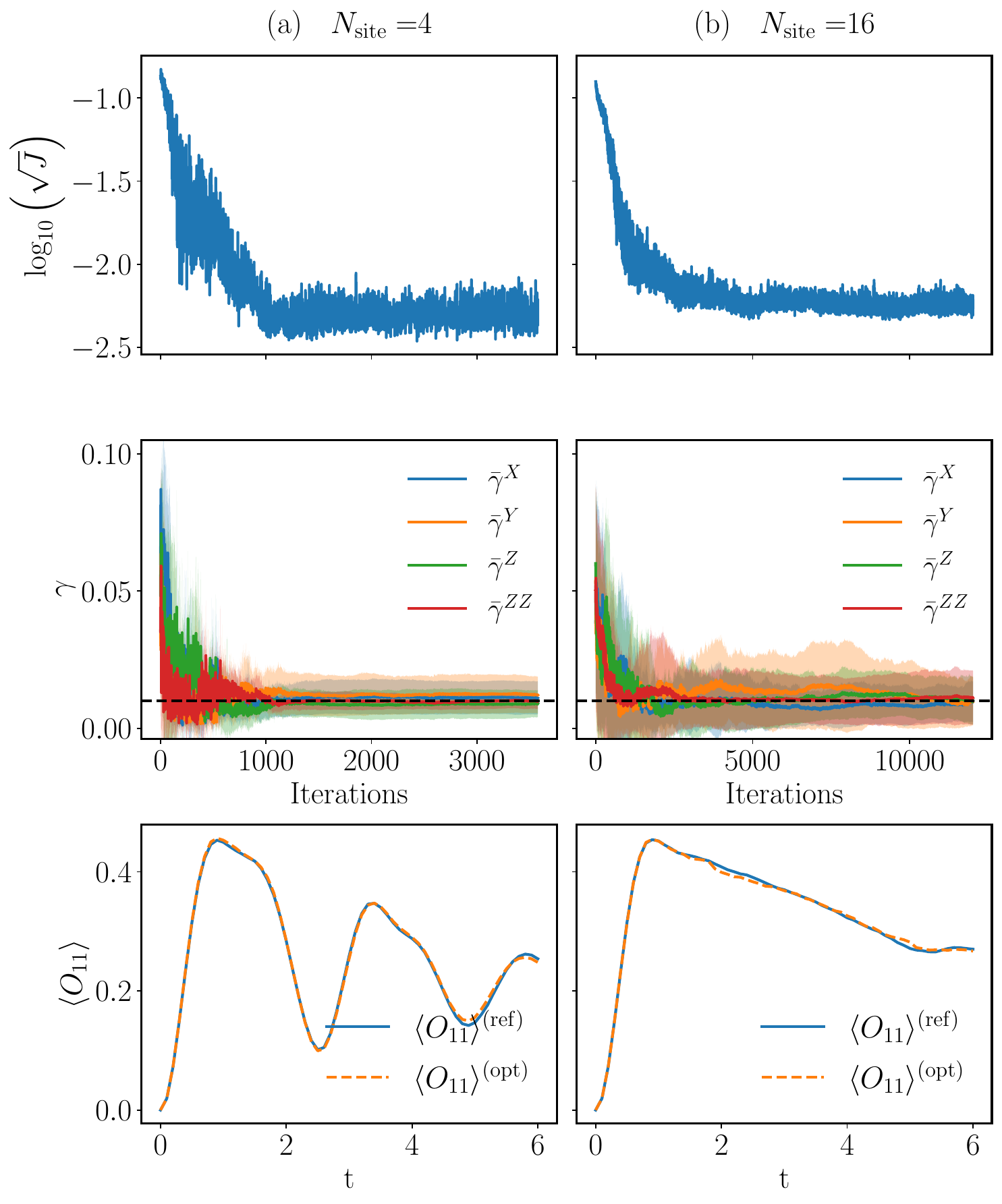}
    \caption{First row: evolution of the cost-function $\sqrt{J}$ during optimization. Second row: per-Pauli-type site-averaged dissipation rate $\bar{\gamma}_j = \frac{1}{N_{\mathrm{site}}}\sum_i \gamma_i^{(j)}$ and pair-averaged crosstalk rate $\bar{\gamma}^{ZZ} = \frac{1}{|\{(i,k)\}|}\sum_{i,k} \gamma_{ik}^{ZZ}$, with the shaded band indicating $\pm$ one standard deviation across sites (pairs); dashed lines mark the reference values $\gamma_i^{(j)\,\mathrm{(ref)}}$ and $\gamma_{ik}^{ZZ\,\mathrm{(ref)}}$. Third row: comparison of reference and optimized trajectories for the observable $\langle O_{11} \rangle = \langle X_1 \rangle$. Column~(a): $N_{\mathrm{site}}=4$; column~(b): $N_{\mathrm{site}}=16$.}
    \label{fig:site_comp_local_noise}
\end{figure*}

In the limit $\epsilon \ll 1$, this expression reduces to:
\begin{equation*}
    \V[\hat{J}] \leq \frac{ (2\ell+1)\,\overline{c\sigma^2} }{N_{\mathrm{site}}}~.
\end{equation*}
The per-observable variances $\sigma_{ijk}^2$ can be
bounded in terms of the purity of the single-site reduced state. The per-observable
variance is defined as
\begin{equation}
\sigma_{ijk}^2
= \mathbb{V}\!\left[\mathrm{Tr}\!\left[\hat\rho_{N_\mathrm{traj}}(t_k)\,O_{ij}\right]\right].
\end{equation}
Since the TJM estimator~\eqref{eq:rhoN_estimator} is the sample mean over
$N_\mathrm{traj}$ independent and identically distributed trajectories, this
reduces to a single-trajectory variance,
\begin{equation}
\sigma_{ijk}^2
= \frac{1}{N_\mathrm{traj}}\,
\mathbb{V}\!\left[\langle\Psi(t_k)|O_{ij}|\Psi(t_k)\rangle\right].
\end{equation}
Since $O_{ij} = I^{\otimes(i-1)}\otimes\sigma_j\otimes I^{\otimes(N_\mathrm{site}-i)}$ acts as
the identity on every site but $i$, Lemma~\ref{lem:local_purity} applies with $\sigma_j$
playing the role of the single-site operator $O$. The single-trajectory variance is then
bounded in terms of the purity $P_i(t_k) := \mathrm{Tr}[\rho_i(t_k)^2]$ of the reduced
single-site state $\rho_i(t_k) := \mathrm{Tr}_{\ne i}[\rho(t_k)]$,
\begin{equation}
\sigma_{ijk}^2 \;\le\; \frac{\|\sigma_j^{\circ}\|_F^2\,\big(1 - P_i(t_k)\big)}{N_\mathrm{traj}} ,
\label{eq:sigma_purity_bound}
\end{equation}
with $\sigma_j^{\circ}$ the traceless part of $\sigma_j$, both computed on the
$2$-dimensional Hilbert space of site $i$ alone. Since the norm entering the bound is
that of the \emph{single-site} operator rather than of its $N_\mathrm{site}$-qubit
embedding $O_{ij}$, this argument applies unchanged to any single-site observable (Pauli
or not) and, more generally, to any observable supported on a fixed number of sites (not
growing with $N_\mathrm{site}$), with $\rho_i(t_k)$ replaced by the reduced state on that
support; only observables whose support itself grows with $N_\mathrm{site}$ would require
a different argument. This bound vanishes as the marginal
purifies ($P_i \to 1$), i.e.\ in the short-time, weak-dissipation regime. Like the
full-system purity of Corollary~\ref{cor:purity_estimator}, $P_i(t_k)$ is directly
estimable from the sampled trajectories. Substituting
Eq.~\eqref{eq:sigma_purity_bound} into $\overline{c\sigma^2}$ then yields
\begin{equation}
\mathbb{V}[\hat{J}] \leq
\frac{2\ell+1}{N_\mathrm{traj} N_\mathrm{site}}\,\frac{1}{\nd}\sum_{ijk} c_{ijk}\|\sigma_j^{\circ}\|_F^2(1 - P_i(t_k)) ,
\end{equation}
and hence the bound for the standard deviation
\begin{equation}
\sigma(\hat{J}) \leq \sqrt{\frac{C}{N_\mathrm{traj} N_\mathrm{site}}} ,
\label{eq:sigma_J_bound}
\end{equation}
with
\begin{equation}
C = \frac{2\ell+1}{\nd}\sum_{ijk} c_{ijk}\,\|\sigma_j^{\circ}\|_F^2\,\big(1 - P_i(t_k)\big) .
\label{eq:C_def}
\end{equation}
For the Pauli matrices used here, $|X_{ijk}| = |\langle O_{ij}\rangle| \le 1$ gives $M_{ijk} = 1$, while $\mathrm{Tr}[\sigma_j] = 0$ gives $\sigma_j^{\circ} = \sigma_j$ and $\|\sigma_j^{\circ}\|_F^2 = \mathrm{Tr}[\sigma_j^2] = 2$. The constant~\eqref{eq:C_def} then reduces to
\begin{equation}
C = 2\,(2\ell+1)\,\overline{c(1-P)} ,
\end{equation}
with $\overline{c(1-P)} := \frac{1}{\nd}\sum_{ijk} c_{ijk}(1-P_i(t_k))$ the average of $c_{ijk}(1-P_i(t_k))$ over all sites, observables, and time steps.

The minimum number of trajectories required to achieve a target standard deviation $\sigma_\mathrm{tar}$ satisfies:
\begin{equation}\label{eq:ntraj_bound}
    N_{\mathrm{traj}} \geq \frac{C}{\sigma_\mathrm{tar}^2 N_{\mathrm{site}}}~.
\end{equation}
This has a direct practical implication: the required number of trajectories decreases inversely with system size, so larger systems can be simulated with less number of trajectories than smaller ones for the same target accuracy.

The finite covariance distance assumption underlying this result is
supported empirically by Fig.~\ref{fig:cov_matrix_ntraj_1000}, which
displays $|\mathrm{Cov}(Y_{ijk}, Y_{i'j'k'})|$ with $j = j' = 1$
(observable $X$) and $k = 37$, $k' = 41$, estimated
for a system size of $N_\mathrm{site} = 40$ using $N_\mathrm{traj} =
1000$. A maximum covariance distance of about 4 sites is visible in
the figure, beyond which the covariance decays rapidly. The
dependence of the maximum covariance distance $\ell$ on the dissipation
strength $\boldsymbol{\gamma}$ and on time is consistent with the
purity dynamics derived in Sec.~\ref{sec:exact_frob_var_tjm}. In
the weak-dissipation limit $\boldsymbol{\gamma} \to 0$ the purity
remains close to one, which suppresses the factor $\overline{c(1-P)}$ in $C$;
at the same time the maximum covariance distance $\ell$ may grow, so the two
effects compete in $C = 2(2\ell+1)\,\overline{c(1-P)}$. Should $\ell$ grow without
bound, the finite covariance distance assumption itself breaks down, and
Eq.~\eqref{eq:sigma_J_bound} must be applied with care.

\section{Results}\label{sec:results}
We evaluate our noise-learning framework on the Ising model introduced in Sec.~\ref{sec:model_sys} under two complementary noise models, presented in Sec.~\ref{sec:local_noise} and Sec.~\ref{sec:global_noise}, respectively.
Their contrasting parameter counts serve two distinct purposes: the local model's high dimensionality stresses the optimizer and lets us cross-validate TJM against an independent, non-tensor-network Monte Carlo solver at small system sizes, while the global model's fixed, small parameter count removes the dimensionality bottleneck and lets us probe scalability up to $N_{\mathrm{site}}=160$ sites.
Together, these two settings test the two main claims of this work: that the learned model reproduces the reference dynamics even when individual rates are not uniquely identifiable, and that the TJM-based pipeline scales to system sizes far beyond the reach of exact simulation.

All simulations were run on nodes equipped with two Intel Xeon Gold 6338 CPUs (32 cores each, 64 physical cores and 128 threads per node), with a peak memory usage of $\sim 13$\,GB for the largest system size of $N_{\mathrm{site}}=160$.
Since the trajectories arising from the unraveling of the density matrix are independent of one another, the large number of available cores allows us to parallelize their computation, distributing all trajectories across multiple cores.
The TJM implementation used throughout this work is publicly available as part of the MQT-YAQS package \cite{ sander_2026_yaqs, wille_2024_mqt}.

\subsection{Local Noise Model}\label{sec:local_noise}
A local noise model is characterized by independent dissipation rates for each site and for each correlated pair.
Specifically, $N_{\mathrm{site}}\times N_{\mathrm{jump}}$ single-site rates $\gamma_i^{(j)}$ are combined with one rate $\gamma_{ik}^{ZZ}$ per pair $(i,k)$ satisfying $1\le k-i\le r_{\max}=4$, yielding $\sum_{r=1}^{\min(4,\,N_{\mathrm{site}}-1)}(N_{\mathrm{site}}-r)$ crosstalk parameters.
The total parameter count is therefore $7N_{\mathrm{site}}-10$ for $N_{\mathrm{site}}\ge 5$.
For this type of noise model, the number of optimized parameters grows linearly with the number of sites, making the optimization in high dimensions challenging for gradient-free optimizers.
For this reason, we generate reference time-dependent expectation values with $\gamma_i^{(j)\,\mathrm{(ref)}} = \gamma_{ik}^{ZZ\,\mathrm{(ref)}} = 0.01$ for relatively small system sizes, $N_{\mathrm{site}}=\{2,~4,~8,~16\}$.
The number of trajectories is chosen according to Equation~\eqref{eq:ntraj_bound} so a target standard deviation of $\sigma_{\mathrm{tar}}=1.5\times10^{-5}$ for the cost-function is achieved, giving $N_{\mathrm{traj}} = \{10929,~5465,~2733,~1367\}$ for the respective system sizes.
For local-noise-model optimization, we choose CMA-ES because it is well-suited to high-dimensional parameter spaces.
In contrast, BO methods are known to scale poorly with increasing dimensionality, both in sample efficiency and computational cost, making them less suitable for this model.
For the simulation method, we compare TJM with the \textit{mcsolve} \cite{qutip_mcsolve} routine in QuTiP \cite{johansson_2012_qutip}, which is a Monte Carlo non-tensor-network solver.
In the optimizations that use QuTiP as the solver, the reference expectation values are also computed with QuTiP using the same number of trajectories as with TJM.

The optimization results are shown in Figure \ref{fig:loss_local_noise}, where we plot the optimal value of the cost-function against the number of sites.
Both solvers achieve small cost-function values, with QuTiP consistently yielding lower values than TJM, and both methods exhibiting a slight tendency for the cost-function to increase with $N_{\mathrm{site}}$.
For $N_{\mathrm{site}}=16$, QuTiP could not be run due to memory limitations, however this size was handled easily by TJM. 

\begin{figure}[t]
    \centering
    \includegraphics[width=\plotsize\linewidth]{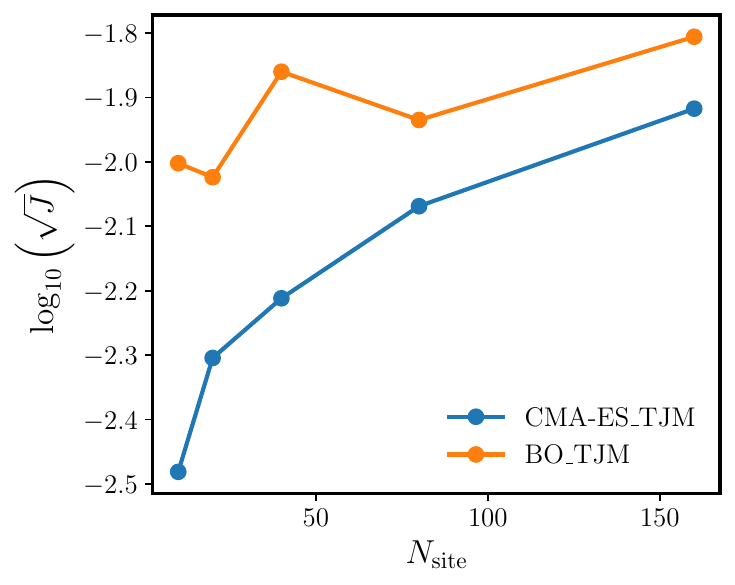}
    \caption{Optimal cost-function value vs. number of sites for a global-noise-model optimization. Both curves are obtained using TJM as Lindblad solver with CMA-ES (blue) and BO (orange) as optimizers.}
    \label{fig:loss_global_noise}
\end{figure}

In Figure \ref{fig:site_comp_local_noise}, we compare the optimization evolution for (a) $N_{\mathrm{site}}=4$ and (b) $N_{\mathrm{site}}=16$ using CMA-ES with TJM.
The last row shows a comparison between the reference trajectory and the trajectory generated using the optimized values of $\gamma_i^{(j)}$ and $\gamma^{ZZ}_{ik}$.
For both system sizes, the cost-function decreases rapidly in the early iterations and subsequently stabilizes at a nearly constant value with small fluctuations, with $N_{\mathrm{site}}=4$ reaching a slightly lower optimal cost compared to $N_{\mathrm{site}}=16$.
For $N_{\mathrm{site}}=4$, the dissipation rates converge to values close to the reference $\gamma_i^{(j)\,\mathrm{(ref)}}$ and $\gamma_{ik}^{ZZ\,\mathrm{(ref)}}$.
For $N_{\mathrm{site}}=16$, the situation is different: the rates do not fully stabilize and exhibit considerable spread across sites, suggesting that the high-dimensional cost-function landscape is relatively flat around the minimum.
Nevertheless, the optimized trajectory agrees closely with the reference one in both cases, indicating that the optimization successfully captures the system dynamics even when the individual rates are not uniquely identified.

\begin{figure*}[t]
    \centering
    \includegraphics[width=0.85\textwidth]{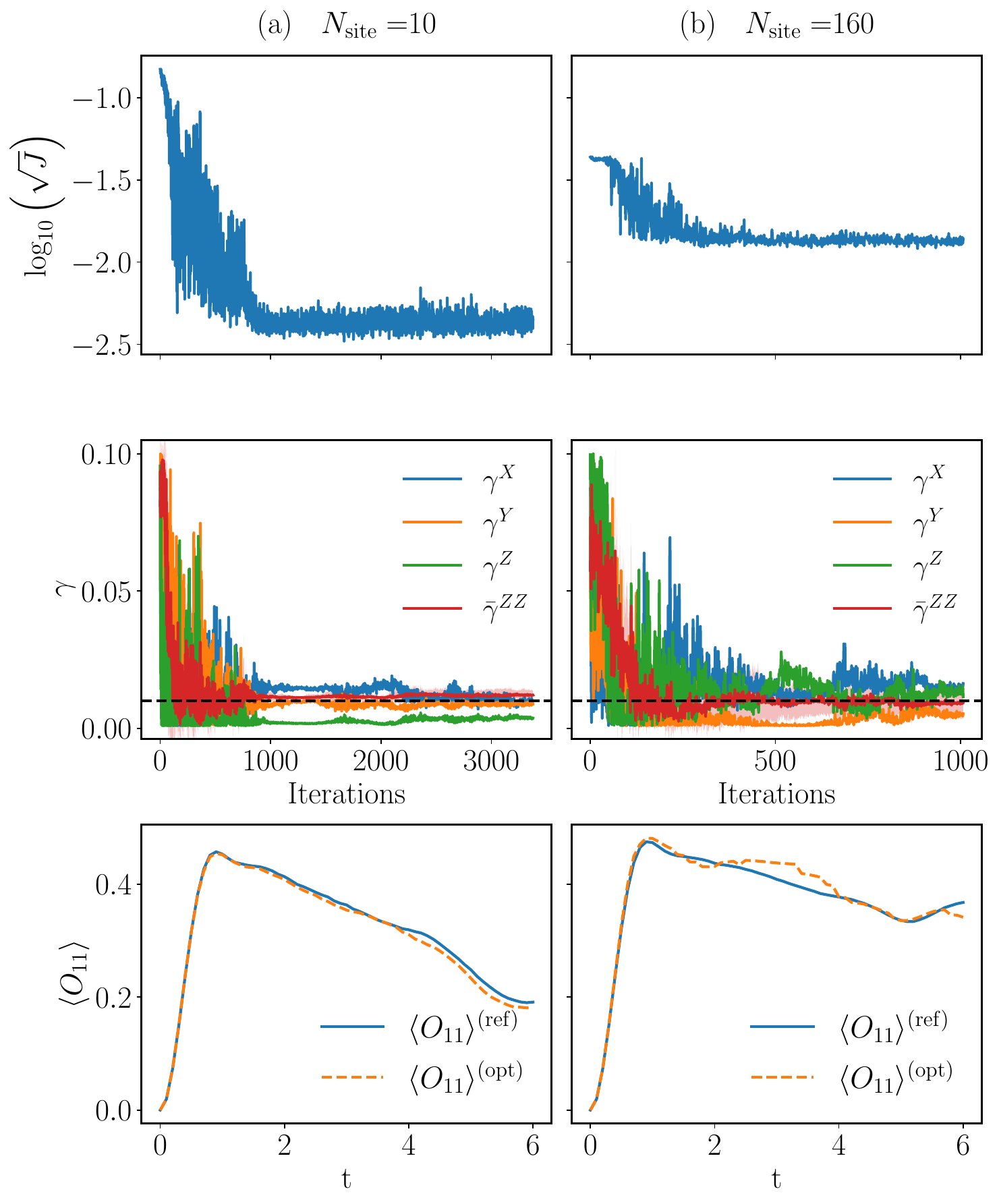}
    \caption{First row: evolution of the cost-function $\sqrt{J}$ during optimization. Second row: evolution of the three global dissipation rates $\gamma_X$, $\gamma_Y$, $\gamma_Z$ and the radius-averaged crosstalk rate $\bar{\gamma}^{ZZ} = \frac{1}{4}\sum_{r=1}^4 \gamma_r^{ZZ}$; dashed lines mark the reference values $\gamma_j^{\mathrm{(ref)}}$ and $\gamma_r^{ZZ\,\mathrm{(ref)}}$. Third row: comparison of reference and optimized trajectories for the observable $\langle O_{11} \rangle = \langle X_1 \rangle$. Column~(a): $N_{\mathrm{site}}=10$; column~(b): $N_{\mathrm{site}}=160$.}
    \label{fig:site_comp_global_noise}
\end{figure*}

\subsection{Global Noise Model}\label{sec:global_noise}
The global noise model assumes that the dissipation rates are spatially homogeneous.
For the single-site Pauli terms, $\gamma_i^{(j)} = \gamma_j$ (one rate per Pauli type), and for the ZZ crosstalk terms, $\gamma_{ik}^{ZZ} = \gamma_r^{ZZ}$ where $r = k-i \in \{1,2,3,4\}$ (one rate per crosstalk radius).
Therefore, only seven parameters are optimized: $\{\gamma_X,~\gamma_Y,~\gamma_Z,~\gamma_1^{ZZ},~\gamma_2^{ZZ},~\gamma_3^{ZZ},~\gamma_4^{ZZ}\}$.
This number is independent of system size, which alleviates the dimensionality challenge and allows us to study larger systems.
Since only seven variables are optimized, BO is a suitable algorithm; in this section, we compare BO with CMA-ES.
In both cases, TJM is used as the Lindblad solver, allowing noise characterization for system sizes $N_{\mathrm{site}}=\{10,~20,~40,~80,~160\}$, with $N_{\mathrm{traj}} = \{2186,~1093,~547,~274,~137\}$ chosen to target the same $\sigma_{\mathrm{tar}}=1.5\times10^{-5}$.
As in the local noise model, we set $\gamma_j^{\mathrm{(ref)}} = 0.01$ for all Pauli types and $\gamma_r^{ZZ\,\mathrm{(ref)}} = 0.01$ for all crosstalk radii $r\in\{1,2,3,4\}$.
The optimization results (analogous to Figure \ref{fig:loss_local_noise}) are shown in Figure \ref{fig:loss_global_noise}.
CMA-ES outperforms BO across all system sizes, likely because the seven-parameter search space already exceeds the regime where BO is most effective.
A decline in performance of CMA-ES with increasing system size is also observed.

In Figure \ref{fig:site_comp_global_noise}, we compare the optimization evolution for (a) $N_{\mathrm{site}}=10$ and (b) $N_{\mathrm{site}}=160$, using CMA-ES/TJM.
For both system sizes, the cost-function decreases rapidly in the early iterations and then stabilizes with small fluctuations, with $N_{\mathrm{site}}=10$ reaching a lower final value than $N_{\mathrm{site}}=160$.
For $N_{\mathrm{site}}=10$, all seven rates, including the four ZZ crosstalk rates, converge to values close to the reference with only small fluctuations, and the optimized trajectory remains close to the reference one.
For $N_{\mathrm{site}}=160$, the rates remain in the vicinity of the reference values but with appreciable fluctuations, and the optimized trajectory reproduces the overall structure of the reference but with noticeable deviations.


\section{Conclusions and Outlook}\label{sec:conclusions}
We introduced a scalable framework for learning Lindblad dissipation rates from time series of local observable expectation values in large-scale open quantum systems.
By combining stochastic simulation techniques with gradient-free optimization of a nonlinear least-squares objective, the approach enables parameter identification in regimes that are out of reach for repeated full-state simulations.
The framework handles both single-site Pauli noise and two-site ZZ correlated dephasing (crosstalk) on all pairs within radius $r_{\max}=4$, and we demonstrated feasibility for system sizes up to $N_{\mathrm{site}}=160$ sites.

Beyond the algorithmic pipeline, we analyzed the statistical fluctuations of the resulting cost-function.
Under a finite covariance distance assumption, we established that the standard deviation of the cost-function decreases proportionally to the inverse square root of the product of the number of sites and trajectories.
This implies that the number of stochastic trajectories required to achieve a fixed target standard deviation decreases proportionally to the inverse of the system size.
As a result, larger systems can be treated more efficiently than with a fixed number of trajectories across all system sizes.

We further proved that, in the full-bond-dimension limit, the Frobenius variance of the TJM density-matrix estimator is exactly $(1 - \mathrm{Tr}[\rho(t)^2])/N_\mathrm{traj}$, providing an exact, purity-based characterization of the stochastic estimation error together with an unbiased empirical diagnostic.

Several extensions are immediate.
On the modeling side, the framework can be generalized to even larger families of jump operators (beyond ZZ pairs), to spatially varying or structured rates, and to settings with partial observability and realistic measurement noise.
On the optimization side, schemes with access to derivatives of the cost-function could substantially improve convergence.
Finally, a natural next step is to integrate the method with experimental data and to systematically study identifiability and confidence intervals, thereby turning scalable Lindblad-rate learning into a practical tool for characterizing dissipation in near-term quantum platforms.

\bibliography{my_library}

\appendix

%
\section{Auxiliary Results}\label{appendix_lemmas}
\begin{lemma}\label{lem:var_y_bound}
    If a random variable $X$ is bounded, $(\exists ~M \in \mathbb{R}:~|X|\leq M)$, and we know $\V[X] = \sigma^2$ and $\E[X] = \mu$, then the variance of the random variable $Y = (X-\mu^{\mathrm{(ref)}})^2$ is bounded by:

    \begin{equation*}
        \V[Y] \leq \sigma^2 c~,
    \end{equation*}
    where $c = \left( M  + |\mu| + 2 \left| \mu - \mu^{\mathrm{(ref)}} \right|\right)^2 $.

\end{lemma}

\begin{proof}
\begin{equation*}
    \V[Y] = \E\left[\left(Y -\E[Y]\right)^2\right]
\end{equation*}
If we define $\delta = \mu - \mu^{\mathrm{(ref)}}$,
\begin{align*}
   \V[Y] & \leq \E\left[\left(Y - \delta^2\right)^2\right], \text{because $\E[Y]$ is the value} \\
    &  ~~~~~~~~~~~~~~~~~~~~~~~~~\text{that minimizes $\E\left[\left(Y -\E[Y]\right)^2\right]$} \\
    & = \E\left[ \left( \left( X - \mu^{\mathrm{(ref)}}\right)^2  - \delta^2 \right)^2\right] \\
    & = \E\left[ \left( X - \mu^{\mathrm{(ref)}} - \delta \right)^2 \left( X - \mu^{\mathrm{(ref)}} + \delta \right)^2 \right] \\
    & = \E\left[ \left( X-\mu \right)^2 \left( X + \mu - 2\mu^{\mathrm{(ref)}} \right)^2 \right]\\
    & \leq \E\left[ \left( X-\mu \right)^2 \right] \left( M  + |\mu| + 2 \left| \mu - \mu^{\mathrm{(ref)}} \right|\right)^2 \\
    & = \sigma^2 \left( M  + |\mu| + 2 \left| \mu - \mu^{\mathrm{(ref)}} \right|\right)^2~,
\end{align*}
where the inequality writes $X + \mu - 2\mu^{\mathrm{(ref)}} = (X-\mu) + 2\delta$ and uses $|X-\mu| \leq |X| + |\mu| \leq M + |\mu|$, so that $\left| X + \mu - 2\mu^{\mathrm{(ref)}} \right| \leq M + |\mu| + 2\left| \mu - \mu^{\mathrm{(ref)}} \right|$.
\end{proof}

\begin{lemma}\label{lem:sum_eps}
Given a matrix $A \in \mathbb{R}^{N \times N}$,
\begin{equation*}
    \qquad
A_{ij} =
\begin{cases}
a, & |i-j| \le \ell, \\
\epsilon^{|i-j|-\ell}, & |i-j| > \ell,
\end{cases}
\end{equation*}
where $\ell \in \mathbb{N}$, $\ell\leq N-1$, and $0 < \epsilon < 1$, then

\begin{align*}
    \sum_{i=1}^{N}\sum_{j=1}^{N} A_{ij} &= a\left[N(2\ell+1) - \ell(\ell+1)\right] \\
    &~~~~~~~+ \frac{2\epsilon}{(1-\epsilon)^2 }\Big( N + \ell\epsilon + \epsilon^{N-\ell} \\
    &~~~~~~~~~~~~~~~~~~~~~~~~~~- N\epsilon - \ell - 1  \Big)~.
\end{align*}

\begin{proof}
    \begin{align*}
         \sum_{i=1}^{N}\sum_{j=1}^{N} A_{ij} &= \sum_{i=1}^{N} a + 2 \sum_{i=1}^{\ell} a(N-i) \\
         &\quad + 2 \sum_{i=1}^{N-\ell-1} (N-\ell-i) \epsilon^i \\
         &=aN+2aN\ell - 2a \sum_{i=1}^{\ell} i + 2(N-\ell)\sum_{i=1}^{N-\ell-1}\epsilon^i \\
         &\qquad -2\sum_{i=1}^{N-\ell-1}i\epsilon^i\\
         &= a\left[N(2\ell+1) - \ell(\ell+1)\right] \\
         &~~~~~~~+ \frac{2\epsilon}{(1-\epsilon)^2 }\Big( N + \ell\epsilon + \epsilon^{N-\ell} \\
         &~~~~~~~~~~~~~~~~~~~~~~~~~~- N\epsilon - \ell - 1  \Big)
    \end{align*}
\end{proof}
\end{lemma}

\begin{lemma}[Single-site variance from local purity]\label{lem:local_purity}
Let $|\Psi(t)\rangle$ be a trajectory sampled by the TJM with
$\E[|\Psi(t)\rangle\langle\Psi(t)|] = \rho(t)$, and let
$\tilde\rho_i(t) := \mathrm{Tr}_{\ne i}\!\left[|\Psi(t)\rangle\langle\Psi(t)|\right]$ be its
reduced density matrix on site $i$, with mean
$\rho_i(t) := \mathrm{Tr}_{\ne i}[\rho(t)] = \E[\tilde\rho_i(t)]$ and single-site purity
$P_i(t) := \mathrm{Tr}[\rho_i(t)^2]$. Let $O$ be any Hermitian operator on the
$2$-dimensional Hilbert space of site $i$, and let
$O_i := I^{\otimes(i-1)}\otimes O\otimes I^{\otimes(N_\mathrm{site}-i)}$ denote its
$N_\mathrm{site}$-qubit embedding, which acts as the identity on every site but $i$.
Then
\begin{equation}
\V\!\left[\langle\Psi(t)|O_i|\Psi(t)\rangle\right] \;\le\; \|O^{\circ}\|_F^2\,\big(1 - P_i(t)\big),
\label{eq:local_purity_bound}
\end{equation}
where $O^{\circ}$ is the traceless part of $O$,
\begin{equation}
O^{\circ} := O - \frac{1}{2}\mathrm{Tr}[O]\,I ,
\end{equation}
and $\|A\|_F^2 := \mathrm{Tr}[A^\dagger A]$.
\end{lemma}

\begin{proof}
Since $O_i$ acts as the identity on every site but $i$, $\langle\Psi|O_i|\Psi\rangle = \mathrm{Tr}[O\,\tilde\rho_i]$,
a random variable over trajectories with mean $\mathrm{Tr}[O\rho_i]$ and deviation
$\mathrm{Tr}[O\,\Delta\rho_i]$, where $\Delta\rho_i := \tilde\rho_i - \rho_i$ is Hermitian and traceless
(a difference of unit-trace density matrices). Because $\Delta\rho_i$ is traceless,
$\mathrm{Tr}[O\Delta\rho_i] = \mathrm{Tr}[O^{\circ}\Delta\rho_i]$, and Cauchy--Schwarz for the
Hilbert--Schmidt inner product gives
$|\mathrm{Tr}[O^{\circ}\Delta\rho_i]| \le \|O^{\circ}\|_F\,\|\Delta\rho_i\|_F$. Hence
\begin{align*}
\V[\langle\Psi|O_i|\Psi\rangle]
&= \E\!\left[(\mathrm{Tr}[O^{\circ}\Delta\rho_i])^2\right] \\
&\le \|O^{\circ}\|_F^2\,\E\!\left[\|\Delta\rho_i\|_F^2\right]
= \|O^{\circ}\|_F^2\,\V_F[\tilde\rho_i].
\end{align*}
Finally, expanding as in Theorem~\ref{thm:exact_frobenius_variance} applied to the
marginal, $\V_F[\tilde\rho_i] = \E[\mathrm{Tr}[\tilde\rho_i^2]] - P_i(t) \le 1 - P_i(t)$, since
$\mathrm{Tr}[\tilde\rho_i^2] \le 1$ for any physical state. Substituting gives
Eq.~\eqref{eq:local_purity_bound}.
\end{proof}

\begin{theorem}[Purity evolution]
\label{thm:purity_evolution}
Let $\rho(t)$ be the solution of the Lindblad master equation
\begin{multline}
\frac{d}{dt}\rho(t)
=
-\ii[H_0,\rho(t)] \\
+
\sum_{m} \gamma_m
\left(
L_m \rho(t) L_m^\dagger
-\frac{1}{2}\{L_m^\dagger L_m,\rho(t)\}
\right).
\label{eq:lindblad_purity}
\end{multline}
Then the purity $P(t)=\mathrm{Tr}[\rho(t)^2]$ satisfies
\begin{multline}
\frac{d}{dt}P(t)
=
2\sum_{m} \gamma_m
\Big(
\mathrm{Tr}[\rho(t)L_m\rho(t)L_m^\dagger] \\
-
\mathrm{Tr}[\rho(t)^2L_m^\dagger L_m]
\Big).
\label{eq:purity_evolution}
\end{multline}
\end{theorem}

\begin{proof}
Differentiating the purity gives
\begin{multline*}
\frac{d}{dt}P(t)
=
\frac{d}{dt}\mathrm{Tr}[\rho(t)^2]
=
\mathrm{Tr}\!\left[\frac{d}{dt}\rho(t)\rho(t)\right] \\
+
\mathrm{Tr}\!\left[\rho(t)\frac{d}{dt}\rho(t)\right]
=
2\,\mathrm{Tr}\!\left[\rho(t)\frac{d}{dt}\rho(t)\right].
\end{multline*}
Substituting Eq.~(\ref{eq:lindblad_purity}) yields
\begin{multline}
\frac{d}{dt}P(t)
=
-2\ii\,\mathrm{Tr}(\rho(t)[H_0,\rho(t)]) \\
+
2\sum_{m} \gamma_m
\mathrm{Tr}\!\Big[
\rho(t)L_m\rho(t)L_m^\dagger \Big. \\
\Big.
-\frac{1}{2}\rho(t)L_m^\dagger L_m\rho(t)
-\frac{1}{2}\rho(t)^2L_m^\dagger L_m
\Big].
\label{eq:purity_derivative_expand}
\end{multline}
The Hamiltonian part vanishes due to cyclicity of the trace:
\begin{equation}
\mathrm{Tr}[\rho[H_0,\rho]]
=
\mathrm{Tr}[\rho H_0 \rho]
-
\mathrm{Tr}[\rho^2 H_0]
=
0.
\label{eq:hamiltonian_part_vanishes}
\end{equation}
Moreover,
\begin{equation}
\mathrm{Tr}[\rho L_m^\dagger L_m \rho]
=
\mathrm{Tr}[\rho^2 L_m^\dagger L_m].
\label{eq:cyclicity_second_term}
\end{equation}
Substituting Eqs.~(\ref{eq:hamiltonian_part_vanishes}) and (\ref{eq:cyclicity_second_term}) into Eq.~(\ref{eq:purity_derivative_expand}) proves Eq.~(\ref{eq:purity_evolution}).
\end{proof}

\begin{proof}[Proof of Corollary~\ref{cor:hermitian_jump}]
Let $L=L^\dagger$ be Hermitian.
Then
\begin{align}
\|[L,\rho]\|_F^2
&=
\mathrm{Tr}\!\left[
(L\rho-\rho L)^\dagger(L\rho-\rho L)
\right] \nonumber \\
&=
\mathrm{Tr}\!\left[
(\rho L-L\rho)(L\rho-\rho L)
\right] \nonumber \\
&=
\mathrm{Tr}[\rho L^2 \rho]
-\mathrm{Tr}[\rho L \rho L] \nonumber \\
&\quad
-\mathrm{Tr}[L\rho L\rho]
+\mathrm{Tr}[L\rho^2L]. \nonumber
\end{align}
By cyclicity of the trace,
\begin{equation*}
\begin{aligned}
&\mathrm{Tr}[\rho L^2 \rho]
=
\mathrm{Tr}[\rho^2 L^2], \\
&\mathrm{Tr}[L\rho^2L]
=
\mathrm{Tr}[\rho^2 L^2],\\
&\mathrm{Tr}[L\rho L\rho]
=
\mathrm{Tr}[\rho L \rho L].
\end{aligned}
\end{equation*}
Hence,
\begin{equation*}
\|[L,\rho]\|_F^2
=
2\,\mathrm{Tr}[\rho^2L^2]
-
2\,\mathrm{Tr}[\rho L \rho L].
\end{equation*}
Using Eq.~(\ref{eq:purity_evolution}) with $L_m=L_m^\dagger$ therefore yields
\begin{align}
\frac{d}{dt}P(t)
&=
2\sum_{m} \gamma_m
\left(
\mathrm{Tr}[\rho L_m \rho L_m]
-
\mathrm{Tr}[\rho^2 L_m^2]
\right) \nonumber \\
&=
-\sum_{m} \gamma_m \|[L_m,\rho]\|_F^2,
\end{align}
which proves the claim.
\end{proof}

\begin{proof}[Proof of Corollary~\ref{cor:short_time}]
Since $\rho(0)^2=\rho(0)$, we have $P(0)=1$.
Evaluating Eq.~(\ref{eq:purity_evolution}) at $t=0$ gives
\begin{align}
\frac{d}{dt}P(0)
&=
2\sum_{m} \gamma_m
\Big(
\mathrm{Tr}[\rho(0)L_m\rho(0)L_m^\dagger] \Big. \nonumber \\
&\quad \Big.
-
\mathrm{Tr}[\rho(0)L_m^\dagger L_m]
\Big). \nonumber
\end{align}
Since $\rho(0)=|\psi_0\rangle\langle\psi_0|$, it follows that
\begin{equation*}
\mathrm{Tr}[\rho(0)L_m\rho(0)L_m^\dagger]
=
|\langle \psi_0|L_m|\psi_0\rangle|^2
\end{equation*}
and
\begin{equation*}
\mathrm{Tr}[\rho(0)L_m^\dagger L_m]
=
\langle \psi_0|L_m^\dagger L_m|\psi_0\rangle.
\end{equation*}
Hence,
\begin{multline*}
\frac{d}{dt}P(0)
=
-2\sum_{m} \gamma_m
\Big(
\langle \psi_0|L_m^\dagger L_m|\psi_0\rangle \Big. \\
\Big.
-
|\langle \psi_0|L_m|\psi_0\rangle|^2
\Big).
\end{multline*}
A first-order Taylor expansion around $t=0$ gives
\begin{equation}
P(t)=1-2t\,S+O(t^2),
\qquad
S:=\sum_m \gamma_m\,(\Delta_{\psi_0} L_m)^2.
\end{equation}
Theorem~\ref{thm:exact_frobenius_variance} then gives
$1-P(t)=2t\,S+O(t^2)$, so that
\begin{equation*}
\sigma_F\!\left[\hat\rho_{N_\mathrm{traj}}(t)\right]
= \sqrt{\frac{2t\,S+O(t^2)}{N_\mathrm{traj}}}.
\end{equation*}
Since $S>0$ by hypothesis, we factor out $\sqrt{2t\,S/N_\mathrm{traj}}$ and
expand $\sqrt{1+O(t)}=1+O(t)$, which proves
Eq.~\eqref{eq:short_time_scaling}.
\end{proof}

\end{document}